\documentclass[journal,comsoc]{IEEEtran}
\usepackage{amsfonts,amssymb,graphicx,mathrsfs,amsxtra,amsmath}
\usepackage{color}
\usepackage{bm}
\usepackage{hyperref}
\usepackage{url}

\usepackage{cite}

\usepackage{amsthm}
\theoremstyle{plain}

\newtheorem{theorem}{Theorem}
\newtheorem{proposition}{Proposition}
\newtheorem{corollary}{Corollary}

\theoremstyle{definition}
\newtheorem{definition}{Definition}

\theoremstyle{remark}
\newtheorem{remark}{Remark}

\newcommand{\eqdef}{\stackrel{\text{def}}{=}}
\newcommand{\word}[1]{\ensuremath{\boldsymbol{#1}}}
\newcommand{\mat}[1]{\boldsymbol{#1}}
\newcommand{\CC}{\mathscr{C}}

\newcommand{\pr}{\mathrm{Pr}}
\newcommand{\WF}{\mathrm{WF}}
\newcommand{\ub}{\mathrm{ub}}

\newcommand{\cv}{\word{c}}
\newcommand{\ev}{\word{e}}

\newcommand{\mv}{\word{m}}
\newcommand{\xv}{\word{x}}

\newcommand{\zv}{\word{z}}

\newcommand{\Gm}{\mat{G}}
\newcommand{\Hm}{\mat{H}}

\newcommand{\Pm}{\mat{P}}
\newcommand{\Sm}{\mat{S}}
\newcommand{\Am}{\mat{A}}
\newcommand{\Bm}{\mat{B}}

\newcommand{\pk}{\mathsf{pk}}
\newcommand{\sk}{\mathsf{sk}}
\newcommand{\ZZ}{\mat{0}}

\newcommand{\GL}{{\normalfont \textsf{GL}}}

\newcommand{\F}{\mathbb{F}}

\newcommand{\Gp}{\mat{G}_{\rm pub}}

\newcommand{\KG}{\mathsf{KeyGen}}
\newcommand{\Enc}{\mathsf{Encrypt}}
\newcommand{\Dec}{\mathsf{Decrypt}}

\newcommand{\MS}[3]{\mathcal{M}_{#1,#2}\left(#3\right)}

\newcommand{\wt}{\mathsf{wt}}

\newcommand{\rank}{\mathsf{rank}}
\newcommand{\dual}[1]{{#1}^\bot}
\newcommand{\HL}{\Hm_{\mathsf{ldpc}}}
\newcommand{\HM}{\Hm_{\mathsf{mdpc}}}
\newcommand{\oL}{\omega_{\mathsf{ldpc}}}
\newcommand{\oM}{\omega_{\mathsf{mdpc}}}

\title{Cryptanalysis of a public key encryption scheme based on
QC-LDPC and QC-MDPC codes}

\begin{document}

%
%

\author{\IEEEauthorblockN{Vlad Dragoi\IEEEauthorrefmark{1}\IEEEauthorrefmark{2} and Herv\'e Tal\'e Kalachi \IEEEauthorrefmark{3}\\
\IEEEauthorblockA{\IEEEauthorrefmark{1} Faculty of Exact Sciences, ``Aurel Vlaicu'' University of Arad, Romania\\
Email: \IEEEauthorrefmark{1}vlad.dragoi@uav.ro}\\
\IEEEauthorblockA{\IEEEauthorrefmark{2} Normandie University, France; UR, LITIS Mont-saint-aignan, France\\
Email: \IEEEauthorrefmark{2} vlad-florin.dragoi@univ-rouen.fr}\\
\IEEEauthorblockA{\IEEEauthorrefmark{3} ERAL, University of Yaounde 1\\
Email: \IEEEauthorrefmark{3} hervekalachi@gmail.com}
}
}

\maketitle

\begin{abstract}
	This letter presents a cryptanalysis of the modified McEliece cryptosystem recently proposed by Moufek, Guenda and Gulliver \cite{GGM16}. The system is based on the juxtaposition of quasi-cyclic LDPC and quasi-cyclic MDPC codes. The idea of our attack is to find an alternative permutation matrix together with an equivalent  LDPC code  which allow the decoding of any cipher-text with a very high probability. We also apply a recent technique to determine weak keys \cite{BDLO16} for this scheme. The results show that the probability of weak keys is high enough that this variant can be ruled out  as a possible secure encryption scheme.

\end{abstract}


%
%
\IEEEpeerreviewmaketitle

\section{Introduction}

 The problem of finding a practical solution for quantum resistant cryptography has become an urgent issue, mainly for two reasons: the existence of a quantum polynomial time algorithm \cite{S94a} that breaks the actual RSA and ECC solutions and the improvements of classical algorithms for the discrete logarithm in small characteristic \cite{BGJT14}. Among the possible candidates for post-quantum cryptography, code-based cryptography is one of the most promising solutions \cite{BBD08}. It is also one of the oldest public key encryption solutions thanks to McEliece's idea \cite{M78}.
 
In the original paper, McEliece proposed to use binary Goppa codes, which remain unbroken, to date. Other families of algebraic codes were proposed (see \cite{N86},\cite{S94}, \cite{JM96} ect). But they were successfully cryptanalyzed, mainly due to their algebraic structure  \cite{SS92,MS07,CMP14}. Probabilistic codes were also considered: concatenated codes were analyzed in \cite{S94,S98}; convolutional codes were proposed in \cite{LJ12}, but successfully cryptanalyzed in \cite{LT13}. Polar codes were also proposed in \cite{SK14} and cryptanalyzed in \cite{BCDOT16}. 

The emergence of all these attacks on several variants of the McEliece cryptosystem shows the importance and necessity of proposing new masking techniques. For example a code with an efficient decoding algorithm could be masked with the use of an arbitrary code in order to sustain the security of the scheme. This technique was proposed for the first time by Wieschebrink in \cite{W06} to avoid the Sidelnikov-Shestakov attack \cite{SS92}. The idea is to use the juxtaposition of a Generalized Reed-Solomon code and a random code.  This solution became famous since it was also used in the case of Reed-Muller codes \cite{GM13}. But these two variants were broken using the square code technique in \cite{CGGOT14,KO15}. Nonetheless the idea was not totally abandoned since Wang \cite{W16} proposed to use the juxtaposition of a GRS code with a random code and then to multiply the generator matrix of the resulting code with a matrix that linearly combines the columns of the GRS and the random code. The main motivation is to obtain a random like code in the end.
Another technique was introduced in \cite{MT16} where the authors propose to use the ``$u | u+v$'' construction with two GRS codes. This new trend also inspired researchers to use the juxtaposition of LDPC and MDPC codes \cite{GGM16}. This is the variant that we analyze in the present article. 
\paragraph{Our contribution}
We describe  a key recovery attack against the modified McEliece cryptosystem based on the juxtaposition of LDPC and MDPC codes \cite{GGM16}. The idea of our attack is to find an alternative permutation matrix together with an equivalent  LDPC code  which allow the decoding of any cipher-text with a very high probability. 
%
%
\section{Background on Coding Theory}
Throughout the paper we denote by $\F_2$ the finite field with $2$ elements  and $\MS{k}{n}{\F_2}$ the set of $k \times n$ matrices with entries in $\F_2$. The \emph{Hamming weight} $\wt (\xv)$ of a vector $\xv \in \F_2^n$ is the number of non-zero coordinates of $\xv.$
A \emph{binary linear code} $\CC$ of length $n$ and dimension $k$ is a $k$-dimensional linear subspace of $\F_2^n.$  A \emph{generator matrix} of $\CC$ is any $k \times n$ matrix $\Gm \in \MS{k}{n}{\F_2}$ with rows that generate $\CC$. 
The dual $\dual{\CC}$ of $\CC$ is the $n-k$-dimensional linear subspace defined by
\[
\dual{\CC} = \Big \{ \zv \in \F_2^n ~:~ \forall \cv \in \CC, ~ \sum_i c_i z_i = 0 \Big \}.
\]
A parity check matrix of $\CC$ is a generator of $\dual{\CC}.$ 
%
%
\begin{definition}\label{def:ldpc_mdpc}
A $(n,k,\omega)$-code is a binary linear code defined by a $k \times n$ parity-check matrix  ($k < n$) where each row has weight $\omega$. 
When $n\to \infty$ we define
\begin{itemize}
\item
 An \emph{LDPC} code is a $(n,k,\omega)$-code with $\omega =  O\left(1\right)$ \cite{G63}.

\item
 An \emph{MDPC} code is a $(n,k,\omega)$-code  with $\omega =  O\left(\sqrt{n}\right)$ \cite{MTSB13}.
\end{itemize}
\end{definition}

Furthermore we will use the $(n,k,\oL)$ notation, respectively $(n,k,\oM).$ The original decoding technique for LDPC codes is the well-known bit flipping algorithm \cite{G63}. This technique is known to provide an error-correction capability which increases linearly with the length of the code, but decreases with the weight of the parity-checks. Therefore MDPC codes suffer from a degradation in decoding performance, compared with LDPC codes. A full description of the bit flipping algorithm can be found in \cite{G63,MTSB13,CS16a}.

\begin{remark}\label{rem:bit_flip_equiv}
Since the performance and the correctness of the bit flipping algorithm depend on the density of the parity-check matrix,  any equivalent parity-check matrix that respects the density condition enables a correct decoding algorithm.  
\end{remark}

\section{Moufek, Guenda and Gulliver's cryptosystem}
In  \cite{GGM16} Moufek, Guenza and Gulliver proposed a McEliece type cryptosystem, based on LDPC and MDPC codes. The scheme is composed of three algorithms: key generation $\KG(\cdot)$, encryption $\Enc(\cdot)$ and decryption $\Dec(\cdot)$. 

\paragraph{$\KG(n,k,t_1,t_2)=(\pk,\sk)$}
\begin{itemize}
\item Pick a generator matrix $\Gm_1$ of a $(n_1,k,\oL)$ LDPC code $\CC_1$ that can correct $t_1$ errors, and a generator matrix $\Gm_2$ of a$(n_2,k,\oM)$ MDPC code denoted $\CC_2$ that can correct $t_2$ errors.  
\item Pick at random  $\Sm$ in $\GL_k(\F_2)$ and an $n\times n$ permutation matrix $\Pm,$ where $n=n_1+n_2.$ 
\item Compute $\Gp \eqdef \Sm \Gm \Pm,$ where $\Gm=\left(\Gm_1\;|\;\Gm_2\right).$
\item Return 
$
\pk=(\Gp,t_1,t_2) \text{~and~} \sk=(\Sm, \Gm_1, \Gm_2, \Pm ).
$ 
\end{itemize}

\paragraph{$\Enc(\mv,\pk)=\zv$}
\begin{itemize}
\item Randomly generate $\ev = \left(\ev_1 \mid \ev_2 \right) \in \F_2^n$ with $\ev_1 \in \F_2^{n_1},$ $\ev_2 \in \F_2^{n_2},$ $\wt (\ev_1) = t_1$ and $\wt (\ev_2) = t_2$. 
\item Compute $\zv = \mv\Gp + \ev$.
\end{itemize}
\paragraph{$\Dec(\zv,\sk)=\mv$}

\begin{itemize}
\item Compute $\zv^{*}=\zv \Pm^{-1}$ and decode using the bit flipping algorithm for $\CC_1$ and $\CC_2.$ The output is $\mv^* \in \F_2^{k}.$ 
\item Return the message $\mv^* \Sm^{-1}.$
\end{itemize}
\begin{remark}
During the decryption, we have $\zv^{*}= \mv \Sm \left(\Gm_1\;|\;\Gm_2\right) +  \left(\ev_1^* \mid \ev_2^* \right)$  with $\left(\ev_1^* \mid \ev_2^* \right)= \ev \Pm^{-1}$ and $\ev_1^* \in \F_2^{n_1}$. The authors of this scheme propose to use the decoding capability of both LDPC and MDPC codes to find $\ev_1^*$ and $\ev_2^*.$ However, we emphasize that one can obtain $\wt (\ev_1^*) > t_1$ or $\wt (\ev_2^*) > t_2.$ This can imply a failure in the decoding algorithm of $\CC_1$ or $\CC_2$. But both situations cannot occur simultaneously since these would then imply that $\wt (\ev \Pm^{-1}) > t.$ It is also important to remark that obtaining $\ev_1^*$ or $\ev_2^*$ is sufficient to recover $\mv \Sm$ and thus $\mv.$ 
\end{remark}   
 
\section{Cryptanalysis of the Moufek, Guenda and Gulliver's scheme. }
  
  We propose here a key recovery attack against the previous cryptosystem. 
  One of the key points in our attacks is the following proposition. 
  
\begin{proposition}\label{prop:juxtaposition}
A parity check matrix of the public code is 
$$\Hm^\prime= \Hm\Pm = \begin{pmatrix} \HL & \bm{0}\\\bm{0}&\HM \\ \Am&\Bm \end{pmatrix} \Pm,$$
  where $\HL \in \MS{n_1-k}{n_1}{\F_2}$ is the low weight parity-check matrix of the LDPC code, $\HM \in \MS{n_2-k}{n_2}{\F_2}$ the low weight parity-check matrix of the MDPC code,  $\left( \Am \mid \Bm \right) \in \MS{k}{n_1 + n_2}{\F_2}$ a full-rank matrix such that $\Gm_1\Am^t+\Gm_2\Bm^t=\ZZ$ and $\Am,\Bm\ne\ZZ.$ 
\end{proposition}
\begin{proof}
The proof is obvious since we have $\Hm^{\prime}\Gp^t = \ZZ.$
\end{proof} 
This proposition shows the existence of a sufficient number of codewords with weights $\oL$ in the dual of the public code. We represent this set of codewords by the matrix $\Hm^* \in \MS{k^*}{n_1 + n_2}{\F_2},$ each row of $\Hm^*$ being an element of the set and conversely. Since these codewords of weights $\oL$ contain the rows of $\begin{pmatrix} \HL & \bm{0}\end{pmatrix} \Pm,$ we have $\rank{\left( \Hm^* \right)} \geq n_1-k.$ One can expect to have $\rank{\left( \Hm^* \right)} > n_1-k$ but our practical experiments always gave an equality, that is to say $\rank{\left( \Hm^* \right)} = n_1-k$. In the sequel we can suppose that $k^*=n_1-k.$  If $k^* > n_1-k,$ one can select only $n_1-k$ rows of $\Hm^*$ that are linearly independent. 

The attack starts with a search for codewords of weights $\oL$. 
  
  \paragraph{\bf Search of codewords of weights $\oL$ } 
 This step aims to find the rows of the matrix $\Hm^* \in \MS{k^*}{n_1 + n_2}{\F_2}.$ Assuming without loss of generality that an adversary knows the value of the parameter $\oL,$ this step can be achieved by applying any of the ISD variants such as, for example, Dumer's algorithm \cite{D91}. This issue is discussed in detail in Section \ref{sec:complexity}, where a complexity analysis of our attack is given. We emphasize that $\oL$ can be easily guessed during this step by starting with $\oL = 1$ and increasing up to the value that satisfies the condition $\rank{\left( \Hm^* \right)} = n_1 - k.$ 
\begin{remark}\label{exist_of_P*}
At the end of this step we have managed to build the matrix $\Hm^* \in \MS{k^*}{n_1 + n_2}{\F_2}$ that generates the same code as $\left( \HL \mid \ZZ \right) \Pm$ and therefore has $n_2$ zero columns.  
\end{remark}
 
\begin{proposition}
There exists a permutation matrix $\Pm^{*}$ such that $\Hm^* \Pm^* = \left( \Hm_{1}^* \mid \ZZ \right) $ with $\Hm_1^* \in \MS{k^*}{n_1}{\F_2}.$ 
Such a matrix can be computed with complexity $O(n)$ and satisfies
\begin{itemize}
\item $\left(\HL \mid \ZZ \right) \Pm \Pm^* = \left( \HL \Pm_1 \mid \ZZ \right),$  
\item $\Pm \Pm^* = \begin{pmatrix}
\Pm_1 & \ZZ \\
\ZZ & \Pm_2
\end{pmatrix}.$ 
\end{itemize}
$\Pm_1$ and $\Pm_2$ being permutation matrices of sizes $n_1$ and $n_2$ respectively.
\end{proposition}
\begin{proof}
The existence of $\Pm^{*}$ comes directly from Remark \ref{exist_of_P*} which also provides $\left(\HL \mid \ZZ \right) \Pm \Pm^* = \left( \HL \Pm_1 \mid \ZZ \right).$ Furthermore, this last equality with the fact that $\Pm_1$ and $\Pm \Pm^*$ are permutation matrices imply that $\Pm \Pm^* = \begin{pmatrix}
\Pm_1 & \ZZ \\
\ZZ & \Pm_2
\end{pmatrix}$ where $\Pm_2$ is also a permutation matrix. To finish, given the matrix $\Hm^*,$ a matrix $\Pm^{*}$ can be easily computed by identifying the $n_2$ zero columns of $\Hm^*.$      
\end{proof}  
This proposition shows that a cryptanalysis is able to find an alternative permutation matrix $\Pm^*$ together with a parity check matrix $\Hm_{1}^*$ of an equivalent LDPC code that can correct the same number of errors as the secret one. In the sequel, we are going to show that the pair $\left( \Pm^{*} , \Hm_{1}^* \right) $ is sufficient to decode any cipher-text with a high probability.

  \paragraph{\bf Decryption with $\left( \Pm^{*} , \Hm_{1}^* \right)$}
We show here the way to decrypt any cipher-text with $\Pm^*$ and $\Hm^*.$  Let $\ev^* = \ev \Pm^* = \left( \ev^*_1 \mid \ev^*_2 \right)$ and $\zv^*  = \zv \Pm^* = \left(\zv^*_1 \mid \zv^*_2 \right) $ with $\ev^*_1 , \zv^*_1 \in \F_2^{n_1}.$ So
\begin{align*}\zv^* & = \left( \mv\Gp + \ev \right) \Pm^* =  \mv \Sm \left( \Gm_1 \Pm_1  \mid \Gm_2 \Pm_2 \right) +\left( \ev^*_1 \mid \ev^*_2 \right).
\end{align*} 
This implies that $\zv^*_1 = \mv \Sm  \Gm_1 \Pm_1 + \ev^*_1.$ Since $\Sm  \Gm_1 \Pm_1$ generates an LDPC code with parity-check matrix $\Hm_{1}^*,$ we can recover $\ev^*_1$ using the bit-flipping algorithm, assuming that $\wt (\ev^*_1) \leq t_1.$        

In the next paragraph we prove that the probability that $\wt(\ev_1^*) \le t_1$ is asymptotically close to 1 when the length of the codes goes to infinity. 
  
   \begin{theorem}\label{thm:proba_t1}
 For $i\in\{1,2\}$ let $n_i,t_i$ be integers such that $t_i<n_i$ with $n_1=\gamma n_2$ and $t_1>\gamma t_2,$ where $\gamma\ge 1.$ Let $\xv=(\xv_1 \mid \xv_2)$ be a random vector over $\F_2^{n_1+n_2}$ with $\wt(\xv)=t_1+t_2,$ where $\xv_i\in \F_2^{n_i}$ for $i \in \{1,2\}.$ Then we have
\[\pr(\wt(\xv_1)> t_1)< t_2 \dfrac{\binom{n_1}{t_1+1}\binom{n_2}{t_2-1}}{\binom{n_1+n_2}{t_1+t_2}}.\]
  \end{theorem}
  
  \begin{proof}First notice that we have
  \begin{equation}
  \forall \;0\le i\le t_1+t_2 \;,\;\pr(\wt(\xv_1)= i)=\dfrac{\binom{n_1}{i}\binom{n_2}{t_1+t_2-i}}{\binom{n_1+n_2}{t_1+t_2}}.
  \label{eq:proba}\end{equation}
  
Using the latter probability we obtain by a simple computation $$\dfrac{\pr(\wt(\xv_1)=i)}{\pr(\wt(\xv_1)=i+1)}=\dfrac{(i+1)(n_2-t_1-t_2+i+1)}{(n_1-i)(t_1+t_2-i)}.$$

Hence for any $i\ge t_1+1$ we obtain  $\dfrac{\pr(\wt(\xv_1)=i)}{\pr(\wt(\xv_1)=i+1)}\ge \dfrac{(t_1+2)(n_2-t_2+2)}{(n_1-t_1)t_2}.$ Replacing $n_1=\gamma n_2$ and $t_1>\gamma t_2$ in the latter fraction we obtain that   \begin{equation}\forall \; i\ge t_1+1\;\text{ we have }\;{\pr(\wt(\xv_1)=i)}>{\pr(\wt(\xv_1)=i+1)}\label{eq:ineg_prob}\end{equation}

From \eqref{eq:ineg_prob} and \eqref{eq:proba} we deduce the desired result.
  \end{proof}

 When $t_1$ and $t_2$ are linear in the code length we obtain the following asymptotic approximation
  \begin{corollary}\label{cor:asympt_proba}
  Let $n_2=n,$ $n_1=\gamma n$ and  $t_1=\alpha n_1,$  $t_2=\beta n_2$ with $\beta<\alpha\gamma\le 1/2.$ Then when $n\to\infty$ we have
  \[\pr(\wt(\xv_1)> t_1)< c_{\alpha,\beta,\gamma}\sqrt{n} 2^{-n\left((\gamma+1)h(\alpha+\frac{\beta-\alpha}{\gamma+1})-\gamma h(\alpha)-h(\beta)\right)},\]
  where $c_{\alpha,\beta,\gamma}$ is a constant and $h$ is the binary entropy function.
  \end{corollary}
  \begin{proof}Apply the Stirling approximation for factorials and expand the series to obtain the result.
  \end{proof}
  \begin{table}[!ht]
\caption{ \tiny{The probability $\pr(\wt(\xv_1)> t_1)$ for $t_1=n_1/20$ and $t_2=n_2/40$ with $n_1=4k$ and $n_2=2k.$ The first row is the exact value for the probability and the second row is the upper bound from Theorem \ref{thm:proba_t1}, namely $\ub_{k,t_1,t_2}=t_2 {\binom{n_1}{t_1+1}\binom{n_2}{t_2-1}}/{\binom{n_1+n_2}{t_1+t_2}}$. The third row is the asymptotic value of the upper bound from Corollary \ref{cor:asympt_proba}, i.e. $\ub_{k,\alpha,\beta,\gamma}=\sqrt{n} 2^{-n\left((\gamma+1)h(\alpha+\frac{\beta-\alpha}{\gamma+1})-\gamma h(\alpha)-h(\beta)\right)}$.} }
\begin{center}
\resizebox{0.8\columnwidth}{!}{  
\begin{tabular}{|c||ccc |}
\hline
k&6851&8261&9857\\
\hline
\hline
$\log_2\left(\pr(\wt(\xv_1)> t_1)\right)$&-117&-141&-166\\
\hline
$\log_2\left(\ub_{k,t_1,t_2}\right)$ Thm\ref{thm:proba_t1}&-110&-133&-159\\
  \hline
$\log_2\left(\ub_{k,\alpha,\beta,\gamma}\right)$ Cor\ref{cor:asympt_proba}  $$&-106&-129&-155\\
  \hline
\end{tabular}
}
\end{center}
\label{fig:0}
\end{table}  
  \section{Complexity analysis and numerical results.}\label{sec:complexity}
\paragraph{Complexity analysis}
The work factor of computing an alternative private key $(\Pm^*,\Hm^*)$ is given by the work factor of the low weight codewords search algorithm plus the computation required to find $\Pm^*.$ The first step can be done using any of the ISD variants such as 
   Dummer (D-ISD,\cite{D91}),  May, Meurer and Thomae (MMT-ISD,\cite{MMT11}) or  Becker, Joux, May and Meurer (BJMM-ISD,\cite{BJMM12}) or May and Ozerov (MO-ISD \cite{MO15}). These algorithms have a time complexity equal to $O(e^{-\omega \ln(1-k/n)(1+o(1))}),$ as long as $\omega=o(n)$ when $n\to \infty$ \cite{CS16}. Hence, in our case computing the matrix $\Hm^*$ requires a work factor asymptotically equal to $O(e^{-\oL \ln(k/(n_1+n_2))(1+o(1))}),$ when $n_1+n_2$ tends to infinity.
Regarding the complexity of computing $\Pm^*,$ it requires $n_1+n_2$ basic operations (here we consider binary additions of  length $k$ binary vectors). Thus the complexity of our attack, denoted by $\WF_{\mathcal{A}}(n_1,n_2,k,\oL)$ is in the worst case dominated by the cost of the best ISD variant. 
\paragraph{Numerical results}
We analyze the effective cost of our attack on some practical parameters. Firstly we considered suggested values given in \cite{GGM16}, more precisely $n_1+n_2=16128.$ As for the co-dimension we analyzed three different cases, $n_1+n_2-k\in \{8064,10080,12096\}.$ Hence the complexity of finding the codewords of weight $15$ in the dual of the public code ( i.e. $\oL=15$), using the BJMM-ISD variant equals $2^{6.13},2^{10.97}$ respectively $2^{19,20}.$ The computations were done using a PariGP implementation similar to that in \cite{P10}. The probability $\pr(\wt(\xv_1)> t_1)$ (for these parameters) is given in Table \ref{fig:0}. These results show that the parameters proposed in \cite{GGM16} are too vulnerable to be considered in practice.  

One might generate a more resistant set of parameters for the scheme, which fact we illustrate in Table \ref{fig:2}. However, it is important to consider the weak keys approach,  a recent technique introduced in \cite{BDLO16}, where the authors use the Extended Euclidean algorithm in order to recover a private key given a public key. We compute the probability of weak keys for the Moufek et al. variant with reasonable parameters in Table \ref{fig:2}.  
\begin{table}[!ht]
\caption{ \tiny{The proportion of weak keys and the complexity of our attack against the Moufek et al. McEliece variant using a  $(4k,k,\oL)$ LDPC code and a $(2k,k,\oM)$ MDPC code.} 
}
\resizebox{\columnwidth}{!}{ 
\begin{tabular}{|c||c|c|c|}
\hline
$(k,\oL)$&$(6851,36)$&$(8261,44)$&$(9857,54)$\\
\hline\hline
$\log\left(\WF_{\mathcal{A}}(4k,2k,k,\oL)\right)$&$80.9$&$101.6$&$127.4$\\
\hline
Proportion of weak keys&$2^{-7.3}$&$2^{-10.6}$&$2^{-14}$\\
\hline 
\end{tabular}
}
\label{fig:2}
\end{table} 

\begin{remark} 
In the first place we remark that the key size for the \cite{GGM16} scheme is considerably greater than for similar schemes such as \cite{B14,MTSB13}.

Notice from Table \ref{fig:2} that the odds of generating weak keys are too big to imagine that the scheme can be protected against this type of attack. Indeed, we find that the scheme might be secured against the weak keys approach by increasing the values of the $\oL.$ But in order to obtain reasonable secure parameters this solution is equivalent to replacing the LDPC code with an MDPC code and thus is of no interest compared with the MDPC McEliece variant \cite{MTSB13}.  
\end{remark}
\begin{remark}
It is also worth mentioning the recent reaction attacks against the QC-MDPC scheme \cite{GJS16} and the QC-LDPC scheme  \cite{FHSZGJ17}. This technique can be used to recover the structure of the MDPC code in the case of Guenda's et al. variant. Nonetheless we remark from Theorem \ref{thm:proba_t1} and Corollary \ref{cor:asympt_proba} that the weight of $\xv_2$ is likely to be much bigger than the error capacity of the MDPC code. Hence, one might not be able to retrieve the initial message unless it uses the LDPC code. 

However, the reaction attack remains highly interesting in similar constructions, namely in the case of the direct sum or Plotkin sum of LDPC and MDPC codes. Indeed, in these cases attacking the LDPC code with our technique is not sufficient for retrieving the initial message, and thus revealing the structure of the MDPC code is necessary. 
\end{remark}
\section{Conclusion}
We have proposed a successful cryptanalysis of the McEliece variant in \cite{GGM16}. Our attack exploits the structure of the dual of the public code and its complexity is dominated by the low weight search algorithm on this dual. The attack is entirely based on finding the structure of the LDPC code, regardless of the nature of the second code. As a consequence, our result can be applied even if the MDPC code is replaced by another code.

We notice that this variant is also vulnerable to the weak keys approach \cite{BDLO16}, since the proportion of weak keys is not negligible. Hence, one can consider that the McEliece variant \cite{GGM16} is too vulnerable to be practical. We also emphasize that similar constructions, such as the direct sum or Plotkin sum of MDPC and LDPC codes can be attacked by combining our technique with the latest reaction attacks. 
%

\bibliographystyle{plain}

\end{document}